\documentclass[aps,pra,onecolumn]{revtex4-2}
\usepackage[
	bookmarks  = false,
	colorlinks = true,
	linkcolor  = blue,
	urlcolor   = blue,
	citecolor  = blue]{hyperref}
\usepackage{amsmath,amssymb,amsthm,mathtools,graphicx}
\usepackage{braket,bm,bbm}
\theoremstyle{remark}
\newtheorem{theorem}{Theorem}
\newtheorem{lemma}{Lemma}
\newtheorem{corollary}{Corollary}
\newtheorem{definition}{Definition}
\newtheorem{proposition}{Proposition}
\newtheorem{example}{Example}

\newcommand{\appsection}[1]{\section{\MakeUppercase{#1}}}
\DeclareMathOperator{\Tr}{Tr}
\DeclareMathOperator{\PT}{PT}
\newcommand{\ketbra}[2]{\ket{#1}\!\bra{#2}}
\begin{document}
\title{Nonlocal quantum state ensembles and quantum data hiding}
\author{Donghoon Ha}
\affiliation{Department of Applied Mathematics and Institute of Natural Sciences, Kyung Hee University, Yongin 17104, Republic of Korea}
\author{Jeong San Kim}
\email{freddie1@khu.ac.kr}
\affiliation{Department of Applied Mathematics and Institute of Natural Sciences, Kyung Hee University, Yongin 17104, Republic of Korea}
\begin{abstract}
We consider the discrimination of bipartite quantum states and establish a relation between a nonlocal quantum state ensemble and quantum data-hiding processing.
Using a bound on the optimal local discrimination of bipartite quantum states, we provide a sufficient condition for a bipartite quantum state ensemble to be used to construct a quantum data-hiding scheme.
Our results are illustrated with examples in multidimensional bipartite quantum systems.
\end{abstract}
\maketitle
\section{Introduction}
Quantum nonlocality is a key feature of quantum information processing and quantum communication \cite{horo2009,chid2013,brun2014,uola2020}. 
Quantum entanglement is a well-known quantum nonlocal correlation which can be used as crucial resources for quantum key distribution and quantum teleportation \cite{eker1991,benn1993,chit2019}. 
Quantum steering is another quantum nonlocality which can be used in various quantum information-processing tasks such as quantum cryptography, randomness certification, and quantum channel discrimination \cite{uola2020}. 
Recently, it was shown that the nonlocality of quantum coherence can provide
advantages for coherence steerability \cite{stre2017,mond2017,hu2018}.

A quantum nonlocal phenomenon can also arise in discriminating multipartite quantum states.
In general, orthogonal quantum states can be discriminated with certainty,
whereas nonorthogonal quantum states cannot \cite{chef2000,berg2007,barn20091,bae2015}.
However, some multipartite orthogonal quantum states exist that cannot be perfectly discriminated by only \emph{local operations and classical communication} (LOCC) \cite{benn19991,ghos2001,walg2002}.
Moreover, some multipartite nonorthogonal quantum states exist that cannot be optimally discriminated by only LOCC \cite{pere1991,chit2013,chit20142}.

The nonlocality of multipartite quantum state discrimination can be useful in quantum communications \cite{terh20011,divi2002}.
Quantum data hiding is one of the typical applications that use the nonlocality of multipartite quantum state discrimination \cite{terh20011,divi2002}. 
The first quantum data-hiding scheme was shown through the nonlocal quantum state ensemble consisting of Bell states \cite{terh20011}; a classical bit is accessible globally but not locally, and the bit is perfectly hidden asymptotically. 
Several quantum data-hiding schemes were also introduced using sets of quantum states \cite{egge2002,mori2013,pian2014,peng2021,lami2021}.
However, it is still natural to ask which kind of nonlocal quantum state ensemble can be used in quantum data-hiding schemes.

Here, we consider the discrimination of bipartite quantum states and establish a relation between the nonlocal quantum state ensemble and quantum data-hiding processing.
We first introduce a quantity realizing the optimal distinguishability of \emph{partially transposed} (PT) states \cite{ha20221}.
By using this quantity, we show that a bipartite orthogonal quantum state ensemble with small PT-state distinguishability can be turned into a globally distinguishable ensemble with arbitrarily small local distinguishability.
Based on this relation, we provide a sufficient condition for a bipartite quantum state ensemble to be used to construct a quantum data-hiding scheme.
Using our quantum data-hiding scheme, classical data of arbitrary size can be perfectly hidden, asymptotically, if the users are restricted to LOCC measurements, whereas the data can be completely recovered if they can collaborate to use global measurements.
Our results are illustrated by examples in multidimensional bipartite quantum systems.

This paper is organized as follows.
In Sec.~\ref{sec:blmd}, we first review the definitions and properties of minimum-error discrimination in bipartite quantum systems.
We further recall a bound on local minimum-error discrimination and provide some properties which are useful in quantum data hiding. 
In Sec.~\ref{sec:qdh}, we provide a sufficient condition for a bipartite quantum state ensemble to be used to construct a quantum data-hiding scheme.
In Sec.~\ref{sec:ex}, our results are illustrated with examples in multidimensional bipartite quantum systems.
In Sec.~\ref{sec:dsc}, we summarize our results and provide a direction for future works.

\section{Bounds of local minimum-error discrimination}\label{sec:blmd}
For bipartite quantum systems held by two parties, Alice and Bob, a quantum state is described by a density operator $\rho$, that is, a positive-semidefinite operator $\rho\succeq0$ with unit trace $\Tr\rho=1$, acting on a bipartite Hilbert space $\mathcal{H}=\mathbb{C}^{d_{\rm A}}\otimes\mathbb{C}^{d_{\rm B}}$.
A measurement is represented by a positive operator-valued measure $\{M_{i}\}_{i}$, that is, a set of positive-semidefinite operators $M_{i}\succeq0$ acting on $\mathcal{H}$ satisfying 
the completeness relation $\sum_{i}M_{i}=\mathbbm{1}$, where $\mathbbm{1}$ is the identity operator acting on $\mathcal{H}$. 
When a measurement $\{M_{i}\}_{i}$ is performed on a quantum  state $\rho$, the probability of obtaining the measurement outcome with respect to $M_{j}$ is $\Tr(\rho M_{j})$.

In this section, we consider the situation of discriminating $n$ bipartite quantum states $\rho_{0},\ldots,\rho_{n-1}$, where the state $\rho_{i}$ is prepared with the probability $\eta_{i}\in[0,1]$ for $i=0,\ldots,n-1$. We denote this situation as the \emph{quantum state discrimination of the ensemble},
\begin{equation}\label{eq:ens}
\mathcal{E}=\{\eta_{i},\rho_{i}\}_{i=0}^{n-1}.
\end{equation}
In discriminating the states from the ensemble $\mathcal{E}$ using a measurement $\mathcal{M}=\{M_{i}\}_{i=0}^{n-1}$, 
we decide that the prepared state is $\rho_{i}$
if the measurement outcome is from $M_{i}$.
In this case, the average probability of correctly guessing the prepared state is
\begin{equation}\label{eq:pcg}
\sum_{i=0}^{n-1}\eta_{i}\Tr(\rho_{i}M_{i}).
\end{equation}
The \emph{minimum-error discrimination} of $\mathcal{E}$ is to achieve the optimal success probability
\begin{equation}\label{eq:pge}
p_{\rm G}(\mathcal{E})=\max_{\mathcal{M}}\sum_{i=0}^{n-1}\eta_{i}\Tr(\rho_{i}M_{i}),
\end{equation}
where the maximum is taken over all possible measurements \cite{hels1969,bae2013}.

A measurement is called a \emph{LOCC measurement} if it can be realized by LOCC.
When the available measurements are limited to LOCC measurements, we denote the maximum success probability by
\begin{equation}\label{eq:ple}
p_{\rm L}(\mathcal{E})=\max_{\substack{\rm LOCC\\ \mathcal{M}}}\sum_{i=0}^{n-1}\eta_{i}\Tr(\rho_{i}M_{i}).
\end{equation}
From the definitions of $p_{\rm G}(\mathcal{E})$ and $p_{\rm L}(\mathcal{E})$ in Eqs.~\eqref{eq:pge} and \eqref{eq:ple}, respectively, $p_{\rm G}(\mathcal{E})$ is obviously an upper bound of $p_{\rm L}(\mathcal{E})$; therefore,
\begin{equation}\label{eq:ldp}
p_{\rm L}(\mathcal{E})\leqslant p_{\rm G}(\mathcal{E}).
\end{equation}
Moreover, guessing the prepared state as the state with largest prior probability is obviously a LOCC measurement; thus, we have
\begin{equation}\label{eq:inq}
\max\{\eta_{0},\ldots,\eta_{n-1}\}\leqslant p_{\rm L}(\mathcal{E}).
\end{equation}

Now, we consider an upper bound of $p_{\rm L}(\mathcal{E})$ in Eq.~\eqref{eq:ple} based on \emph{partial transposition}. 
We further provide some properties of the upper bound which are useful in quantum data hiding.

\begin{definition}
A Hermitian operator $E$ on a bipartite Hilbert space $\mathcal{H}$ is called \emph{positive partial transpose} (PPT) if 
$E^{\PT}\succeq0$,
where the superscript $\PT$ indicates
the partial transposition \cite{pere1996,pptp}. Otherwise, $E$ is called \emph{negative partial transpose} (NPT).
\end{definition}

For an ensemble $\mathcal{E}=\{\eta_{i},\rho_{i}\}_{i=0}^{n-1}$,
let us consider 
\begin{equation}\label{eq:qge}
q_{\rm G}(\mathcal{E})=\max_{\mathcal{M}}\sum_{i=0}^{n-1}\eta_{i}\Tr(\rho_{i}^{\PT}M_{i}),
\end{equation}
where the maximum is taken over all possible measurements. 
This quantity is known as an upper bound of $p_{\rm L}(\mathcal{E})$; that is,
\begin{equation}\label{eq:upb}
p_{\rm L}(\mathcal{E})\leqslant q_{\rm G}(\mathcal{E})
\end{equation}
for any bipartite quantum state ensemble $\mathcal{E}$ \cite{ha20221}.
The following proposition provides a necessary and sufficient condition of a measurement realizing $q_{\rm G}(\mathcal{E})$ \cite{ha20221}.

\begin{proposition}\label{prop:qgm}
For a bipartite quantum state ensemble $\mathcal{E}=\{\eta_{i},\rho_{i}\}_{i=0}^{n-1}$
and a measurement $\mathcal{M}=\{M_{i}\}_{i=0}^{n-1}$,
$\mathcal{M}$ realizes $q_{\rm G}(\mathcal{E})$ if and only if 
\begin{equation}\label{eq:nsqg}
\sum_{j=0}^{n-1}\eta_{j}\rho_{j}^{\PT}M_{j}-\eta_{i}\rho_{i}^{\PT}\succeq 0
\end{equation}
for each $i=0,\ldots,n-1$.
\end{proposition}

Although Proposition~\ref{prop:qgm} provides a necessary and sufficient condition of the optimal measurement realizing $q_{\rm G}(\mathcal{E})$,
the definition of $q_{\rm G}(\mathcal{E})$ in Eq.~\eqref{eq:qge} is based on the optimization over all possible measurements. However, the following theorem provides the closed-form expression of $q_{\rm G}(\mathcal{E})$ if the ensemble consists of two states.
The proof of Theorem~\ref{thm:tsqg} is given in Appendix~\ref{app:thm12}.

\begin{theorem}\label{thm:tsqg}
For a bipartite quantum state ensemble $\mathcal{E}=\{\eta_{i},\rho_{i}\}_{i=0}^{1}$,
\begin{equation}\label{eq:qgts}
q_{\rm G}(\mathcal{E})=\frac{1}{2}+\frac{1}{2}\Tr|\eta_{0}\rho_{0}^{\PT}-\eta_{1}\rho_{1}^{\PT}|,
\end{equation}
where $|E|$ is the positive square root of $E^{\dagger}E$, that is, 
\begin{equation}\label{eq:abe}
|E|=\sqrt{E^{\dagger}E}.
\end{equation}
\end{theorem}

For an ensemble $\mathcal{E}$ with an arbitrary number of states, the following theorem provides an upper bound of $q_{\rm G}(\mathcal{E})$.
The proof of Theorem~\ref{thm:ubqg} is given in Appendix~\ref{app:thm12}.

\begin{theorem}\label{thm:ubqg}
For a bipartite quantum state ensemble $\mathcal{E}=\{\eta_{i},\rho_{i}\}_{i=0}^{n-1}$ and a Hermitian operator $H$, if 
\begin{equation}\label{eq:ubc}
H-\eta_{i}\rho_{i}^{\PT}\succeq0
\end{equation}
for each $i=0,\ldots,n-1$, then the trace of $H$ is an upper bound of $q_{\rm G}(\mathcal{E})$, that is,
\begin{equation}\label{eq:trh}
\Tr H \geqslant q_{\rm G}(\mathcal{E}).
\end{equation}
\end{theorem}

\section{Quantum data hiding}\label{sec:qdh}
In this section, we show that a bipartite orthogonal quantum state ensemble with small $q_{\rm G}(\mathcal{E})$, the optimal PT-state distinguishability, can be turned into a globally distinguishable ensemble with arbitrary small local distinguishability.
Based on this relation, we provide a sufficient condition for a bipartite quantum state ensemble that can be used to construct a quantum data-hiding scheme. In our scheme, the hidden data is globally accessible, and LOCC measurements can provide arbitrarily little information about the data. 
To analyze how reliably the data is hidden, we first introduce the notion of a \emph{multifold ensemble}.

For positive integers $n$ and $L$, we denote 
$\mathbb{Z}_{n}$ as the set of all integers from $0$ to $n-1$ and $\mathbb{Z}_{n}^{L}$ as the Cartesian product of $L$ copies of $\mathbb{Z}_{n}$. 
For the ensemble $\mathcal{E}=\{\eta_{i},\rho_{i}\}_{i=0}^{n-1}$ in Eq.~\eqref{eq:ens} and each vector
\begin{equation}
\vec{c}=(c_{1},\ldots,c_{L})\in\mathbb{Z}_{n}^{L},
\end{equation}
we define
\begin{equation}\label{eq:erlf}
\eta_{\vec{c}}=\prod_{l=1}^{L}\eta_{c_{l}},~
\rho_{\vec{c}}=\bigotimes_{l=1}^{L}\rho_{c_{l}},
\end{equation}
where $\rho_{c_{l}}$ is the state in the ensemble $\mathcal{E}$ whose index is the $l$th coordinate of the vector $\vec{c}$, and $\eta_{c_{l}}$ is the corresponding probability for $l=1,\ldots,L$.
We further use $\mathcal{E}^{\otimes L}$ to denote the \emph{$L$-fold ensemble} of $\mathcal{E}$, that is,
\begin{equation}\label{eq:lfe}
\mathcal{E}^{\otimes L}=\{\eta_{\vec{c}},\rho_{\vec{c}}\}_{\vec{c}\in\mathbb{Z}_{n}^{L}}.
\end{equation}
For each $\vec{c}\in\mathbb{Z}_{n}^{L}$, we define $\omega_{n}(\vec{c})$ 
as the modulo $n$ summation of all entries in $\vec{c}=(c_{1},\ldots,c_{L})$, that is,
\begin{equation}\label{eq:mns}
\omega_{n}(\vec{c})=\sum_{l=1}^{L}c_{l}~(\mathrm{mod}~n).
\end{equation}
Clearly, we have $\omega_{n}(\vec{c})\in\mathbb{Z}_{n}$ for any $\vec{c}\in\mathbb{Z}_{n}^{L}$.

For the state $\rho_{\vec{c}}$ prepared from the ensemble $\mathcal{E}^{\otimes L}$, let us consider the situation of guessing $\omega_{n}(\vec{c})$ of the prepared state $\rho_{\vec{c}}$ from $\mathcal{E}^{\otimes L}$.
This situation is equivalent to discriminating the states $\rho_{0}^{(L)},\ldots,\rho_{n-1}^{(L)}$ prepared with the probabilities $\eta_{0}^{(L)},\ldots,\eta_{n-1}^{(L)}$, respectively, where
\begin{eqnarray}
\eta_{i}^{(L)}=
\sum_{\substack{\vec{c}\in\mathbb{Z}_{n}^{L}\\ \omega_{n}(\vec{c})=i}}
\eta_{\vec{c}},~
\rho_{i}^{(L)}=\frac{1}{\eta_{i}^{(L)}}
\sum_{\substack{\vec{c}\in\mathbb{Z}_{n}^{L}\\ \omega_{n}(\vec{c})=i}}\eta_{\vec{c}}\rho_{\vec{c}}.\label{eq:rhil}
\end{eqnarray}
In other words, quantum state discrimination of the ensemble 
\begin{equation}\label{eq:eld}
\mathcal{E}^{(L)}=\{\eta_{i}^{(L)},\rho_{i}^{(L)}\}_{i=0}^{n-1}.
\end{equation}

For a given two-state ensemble $\mathcal{E}=\{\eta_{i},\rho_{i}\}_{i=0}^{1}$, the following theorem gives the closed-form expression of $q_{\rm G}(\mathcal{E}^{(L)})$.
The proof of Theorem~\ref{thm:tsege} is given in Appendix~\ref{app:thm3}.

\begin{theorem}\label{thm:tsege}
For a bipartite quantum state ensemble $\mathcal{E}=\{\eta_{i},\rho_{i}\}_{i=0}^{1}$ and any positive integer $L$,
\begin{equation}\label{eq:qge2}
q_{\rm G}(\mathcal{E}^{(L)})=\frac{1}{2}+\frac{1}{2}\big(\Tr|\eta_{0}\rho_{0}^{\PT}-\eta_{1}\rho_{1}^{\PT}|\big)^{L}
=\frac{1}{2}+\frac{1}{2}\big[2 q_{\rm G}(\mathcal{E})-1\big]^{L}.
\end{equation}
Moreover, we have
\begin{equation}\label{eq:plqgl}
p_{\rm L}(\mathcal{E}^{(L)})=q_{\rm G}(\mathcal{E}^{(L)})
\end{equation}
if $p_{\rm L}(\mathcal{E})=q_{\rm G}(\mathcal{E})$.
\end{theorem}

Theorem~\ref{thm:tsege} provides the closed-form expression of $q_{\rm G}(\mathcal{E}^{(L)})$ for the case of the two-state ensemble $\mathcal{E}$.
The following theorem provides an upper bound of $q_{\rm G}(\mathcal{E}^{(L)})$ for a bipartite quantum state ensemble $\mathcal{E}$ with an arbitrary number of states.
The proof of Theorem~\ref{thm:lqge} is given in Appendix~\ref{app:thm4}.

\begin{theorem}\label{thm:lqge}
For a bipartite quantum state ensemble $\mathcal{E}=\{\eta_{i},\rho_{i}\}_{i=0}^{n-1}$ and any positive integer $L$, we have
\begin{equation}\label{eq:qgel}
q_{\rm G}(\mathcal{E}^{(L)})\leqslant\frac{1}{n}+\frac{n-1}{n}\big[n q_{\rm G}(\mathcal{E})-1\big]^{L}.
\end{equation}
\end{theorem}

From Theorem \ref{thm:lqge}, we have
\begin{equation}\label{eq:npln}
\frac{1}{n}\leqslant p_{\rm L}(\mathcal{E}^{(L)})\leqslant q_{\rm G}(\mathcal{E}^{(L)})
\leqslant\frac{1}{n}+\frac{n-1}{n}\big[n q_{\rm G}(\mathcal{E})-1\big]^{L},
\end{equation}
where the first inequality is due to inequality~\eqref{eq:inq}, the second inequality is from inequality~\eqref{eq:upb}, and the last inequality is from Theorem~\ref{thm:lqge}.
For the case of $q_{\rm G}(\mathcal{E})<\frac{2}{n}$,
the last term in Eq.~\eqref{eq:npln} decreases to $\frac{1}{n}$ as $L$ gets larger.
Thus, we have the following corollary.

\begin{corollary}\label{cor:lim}
For a bipartite quantum state ensemble $\mathcal{E}=\{\eta_{i},\rho_{i}\}_{i=0}^{n-1}$ with
$q_{\rm G}(\mathcal{E})<\frac{2}{n}$,
\begin{equation}\label{eq:pln}
\lim_{L\rightarrow\infty}p_{\rm L}(\mathcal{E}^{(L)})=\frac{1}{n}.
\end{equation}
\end{corollary}
We further note that the last term in Eq.~\eqref{eq:npln} decreases to $\frac{1}{n}$ exponentially fast with respect to $L$.
Thus, the convergence of $p_{\rm L}(\mathcal{E}^{(L)})$ to  $\frac{1}{n}$ is also exponentially fast with respect to $L$.

In information theory, data hiding is the protection of a data by sharing it with multiple users so that the data becomes accessible only as a consequence of their cooperative action. 
By using particular bipartite quantum states, several quantum data-hiding schemes have been proposed to share data between two parties, Alice and Bob \cite{terh20011,divi2002,egge2002,mori2013,pian2014,peng2021,lami2021}.
If Alice and Bob can perform only LOCC measurements, they can obtain arbitrarily little information about the hidden data.
They can access the data only by global measurements, which require either a quantum channel, shared quantum entanglement, or direct interaction between them.

\begin{figure*}[!tt]
\centerline{\includegraphics{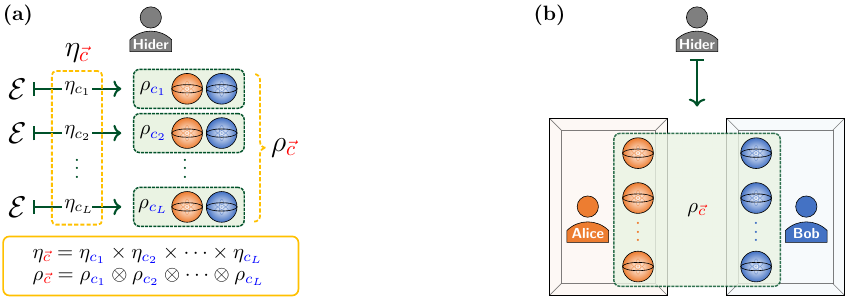}}
\caption{Data-hiding scheme based on a bipartite orthogonal quantum state ensemble $\mathcal{E}=\{\eta_{i},\rho_{i}\}_{i=0}^{n-1}$ with $q_{\rm G}(\mathcal{E})<\frac{2}{n}$
(step~1). (a) Hider first prepares the state $\rho_{\vec{c}}$ with the corresponding probability $\eta_{\vec{c}}$ from the $L$-fold ensemble of $\mathcal{E}$. (b) Hider sends $\rho_{\vec{c}}$ to Alice and Bob. 
}\label{fig:qdh1}
\end{figure*}

Here, we establish a relation between a bipartite quantum state ensemble and its role in quantum data hiding. 
For a given bipartite quantum state ensemble $\mathcal{E}$, we first provide a condition in terms of $p_{\rm G}(\mathcal{E})$ and $q_{\rm G}(\mathcal{E})$ defined in Eqs.~\eqref{eq:pge} and \eqref{eq:qge}, respectively. 
We show that $\mathcal{E}$ can be used to construct a quantum data hiding scheme if $\mathcal{E}$ satisfies the condition.
The data is globally accessible, and LOCC measurements can provide arbitrarily little information about the data. In other words, the data is perfectly hidden asymptotically.

For any integer $n\geqslant 2$, let us consider a bipartite quantum state ensemble $\mathcal{E}=\{\eta_{i},\rho_{i}\}_{i=0}^{n-1}$ satisfying the following condition:
\begin{equation}\label{eq:pqgc}
p_{\rm G}(\mathcal{E})=1,~~q_{\rm G}(\mathcal{E})<\frac{2}{n}.
\end{equation}
By using the ensemble $\mathcal{E}$, we construct a quantum data-hiding scheme that can hide an $n$-ary classical data, that is, one in $\{0,\ldots,n-1\}$.

The hider, the third party, first chooses a bipartite quantum state $\rho_{c_{1}}$ from $\mathcal{E}$ with the corresponding probability $\eta_{c_{1}}$ for $c_{1}\in\mathbb{Z}_{n}$, and shares it with Alice and Bob. 
The hider repeats this process $L$ times, so that $\rho_{c_{l}}$ is chosen with probability $\eta_{c_{l}}$ in the $l$th repetition for $l=1,\ldots,L$.
This whole procedure is equivalent to the situation in which the hider first prepares the state 
\begin{equation}\label{eq:psm}
\rho_{\vec{c}}=\rho_{c_{1}}\otimes\cdots\otimes\rho_{c_{L}}
\end{equation}
from the $L$-fold ensemble $\mathcal{E}^{\otimes L}$ in Eq.~\eqref{eq:lfe}
with the corresponding probability $\eta_{\vec{c}}$ for $\vec{c}=(c_{1},\ldots,c_{L})\in\mathbb{Z}_{n}^{L}$, and sends it to Alice and Bob. 
This step is illustrated in Fig.~\ref{fig:qdh1}.

To hide classical data $x\in\mathbb{Z}_{n}$, the hider broadcasts the classical information $z$ to Alice and Bob,
\begin{equation}\label{eq:xyz}
z=x\oplus y,
\end{equation}
where $\oplus$ is the modulo $n$ addition and
\begin{equation}\label{eq:yonc}
y=\omega_{n}(\vec{c})
\end{equation}
for the prepared state $\rho_{\vec{c}}$ in Eq.~\eqref{eq:psm}.
We illustrate this step in Fig.~\ref{fig:qdh2}. 

\begin{figure}[!tt]
\centerline{\includegraphics{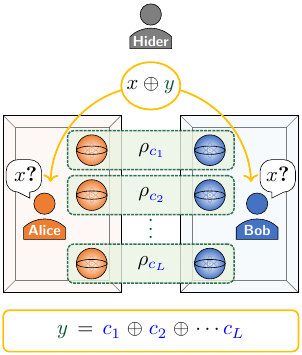}}
\caption{Data-hiding scheme based on a bipartite orthogonal quantum state ensemble $\mathcal{E}=\{\eta_{i},\rho_{i}\}_{i=0}^{n-1}$ with $q_{\rm G}(\mathcal{E})<\frac{2}{n}$
(step~2).
To hide the classical data $x\in\{0,\ldots,n-1\}$, the hider broadcasts $x\oplus y$ to Alice and Bob, where $y$ is the modulo $n$ summation of all entries in $\vec{c}$
for the prepared state $\rho_{\vec{c}}$.}\label{fig:qdh2}
\end{figure}

As the information $z$ in Eq.~\eqref{eq:xyz} was broadcasted to Alice and Bob, guessing the data $x$ is equivalent to guessing the data $y$. 
Moreover, guessing the data $y$ in Eq.~\eqref{eq:yonc} is equivalent to discriminating the states in Eq.~\eqref{eq:rhil} due to the argument in the paragraph containing Eqs.~\eqref{eq:mns} and \eqref{eq:rhil}.
Thus, the maximum average probability of correctly guessing the data $x$ is equal to the maximum average probability of discriminating the states from the ensemble $\mathcal{E}^{(L)}$ in Eq.~\eqref{eq:eld}, that is, $p_{\rm G}(\mathcal{E}^{(L)})$.
Similarly, the maximum average probability of correctly guessing the data $x$ using only LOCC measurements becomes $p_{\rm L}(\mathcal{E}^{(L)})$.

The condition $p_{\rm G}(\mathcal{E})=1$ in Eq.~\eqref{eq:pqgc} means the states $\rho_{i}$ of $\mathcal{E}$ are mutually orthogonal. 
In this case, the states $\rho_{\vec{c}}$ of $\mathcal{E}^{\otimes L}$ in Eq.~\eqref{eq:erlf} are mutually orthogonal, which also implies the mutual orthogonality of $\rho_{i}^{(L)}$ of $\mathcal{E}^{(L)}$ in Eq.~\eqref{eq:rhil}.
Thus, we have
\begin{equation}\label{eq:dpg1}
p_{\rm G}(\mathcal{E}^{(L)})=1.
\end{equation}
In other words, the data $x$ is accessible when Alice and Bob can collaborate to use global measurements.

On the other hand, the condition $q_{\rm G}(\mathcal{E})<\frac{2}{n}$ in Eq.~\eqref{eq:pqgc} and Corollary~\ref{cor:lim} imply that
$p_{\rm L}(\mathcal{E}^{(L)})$ can be arbitrarily close to $\frac{1}{n}$ with the choice of an appropriately large $L$.
Thus, the globally accessible data $x$ is perfectly hidden asymptotically if Alice and Bob can use only LOCC measurements.

\section{Examples}\label{sec:ex}
In this section, we provide examples of bipartite quantum state ensemble $\mathcal{E}$ satisfying condition~\eqref{eq:pqgc} to illustrate our results. 
First, we present the following example of the two-state ensemble $\mathcal{E}=\{\eta_{i},\rho_{i}\}_{i=0}^{1}$.
\begin{example}\label{ex:npt}
For a NPT quantum state $\sigma$, let us consider the two-state ensemble $\mathcal{E}=\{\eta_{i},\rho_{i}\}_{i=0}^{1}$ consisting of 
\begin{flalign}\label{eq:ex1}
\eta_{0}=\frac{\Tr|\sigma^{\PT}|+1}{2\Tr|\sigma^{\PT}|},&~~
\rho_{0}=\frac{|\sigma^{\PT}|+\sigma^{\PT}}{\Tr|\sigma^{\PT}|+1},\nonumber\\
\eta_{1}=\frac{\Tr|\sigma^{\PT}|-1}{2\Tr|\sigma^{\PT}|},&~~
\rho_{1}=\frac{|\sigma^{\PT}|-\sigma^{\PT}}{\Tr|\sigma^{\PT}|-1}.
\end{flalign}
We first note that $\Tr|\sigma^{\PT}|>1$ because $\sigma$ is a NPT state.
\end{example}
It is straightforward to verify
\begin{equation}\label{eq:exq1}
\Tr|\eta_{0}\rho_{0}^{\PT}-\eta_{1}\rho_{1}^{\PT}|=\frac{1}{\Tr|\sigma^{\PT}|}.
\end{equation}
From Theorem~\ref{thm:tsqg} together with Eqs.~\eqref{eq:ex1} and \eqref{eq:exq1}, we have
\begin{equation}\label{eq:eqge}
q_{\rm G}(\mathcal{E})=
\frac{\Tr|\sigma^{\PT}|+1}{2\Tr|\sigma^{\PT}|}
=\eta_{0}.
\end{equation}

We also note that
\begin{equation}\label{eq:exqg1}
\eta_{0}\leqslant p_{\rm L}(\mathcal{E})\leqslant
q_{\rm G}(\mathcal{E})=\eta_{0},
\end{equation}
where the first inequality is from inequality~\eqref{eq:inq}, 
the second inequality is due to inequality~\eqref{eq:upb},
and the last equality is from Eq.~\eqref{eq:eqge}.
Thus, we have
\begin{equation}\label{eq:pleg}
p_{\rm L}(\mathcal{E})=q_{\rm G}(\mathcal{E})=\eta_{0}.
\end{equation}
From Theorem~\ref{thm:tsege} together with Eqs.~\eqref{eq:exq1} and \eqref{eq:pleg}, we have
\begin{equation}\label{eq:eqlv}
p_{\rm L}(\mathcal{E}^{(L)})=q_{\rm G}(\mathcal{E}^{(L)})=\frac{1}{2}+\frac{1}{2(\Tr|\sigma^{\PT}|)^{L}}
\end{equation}
for any positive integer $L$.

Because $|E|^{2}=E^{2}$ for any Hermitian operator $E$, we can easily check that $\rho_{0}$ and $\rho_{1}$ in Eq.~\eqref{eq:ex1} are orthogonal; therefore, the condition $p_{\rm G}(\mathcal{E})=1$ in Eq.~\eqref{eq:pqgc} holds.
Moreover, the condition $q_{\rm G}(\mathcal{E})<1$ in Eq.~\eqref{eq:pqgc} holds due to  Eq.~\eqref{eq:pleg}.
Thus, the ensemble $\mathcal{E}$ in Example~\ref{ex:npt} can be used to construct a data-hiding scheme that hides one classical bit.

Now, we provide another example of an ensemble with an arbitrary number of states in multidimensional bipartite quantum systems.

\begin{example}\label{ex:gex}
For positive integers $d$,\,$m$, and $n$ with $d\geqslant2$ and $2^{m-1}<n\leqslant 2^{m}$,
let us consider the $d^{m}\otimes d^{m}$ quantum state ensemble $\mathcal{E}=\{\eta_{i},\rho_{i}\}_{i=0}^{n-1}$ consisting of $m$-fold tensor products of $d\otimes d$ quantum states $(\mathbbm{1}\pm\mathcal{F})/(d^{2}\pm d)$,
\begin{equation}\label{eq:enx2}
\eta_{i}=\frac{1}{\mathcal{N}}\prod_{k=1}^{m}\big[d^{2}+(-1)^{b_{k}(i)}d\big],~
\rho_{i}=\bigotimes_{k=1}^{m}\frac{\mathbbm{1}+(-1)^{b_{k}(i)}\mathcal{F}}{d^{2}+(-1)^{b_{k}(i)}d},
\end{equation}
where $b_{k}(i)$ is the $k$th digit in the $m$-digit binary representation of $i$,
\begin{equation}\label{eq:bki}
\sum_{k=1}^{m}b_{k}(i) 2^{k-1}=i,
\end{equation}
$\mathcal{F}$ is the flip operator in a two-qudit system,
\begin{equation}\label{eq:foe}
\mathcal{F}=\sum_{i,j=0}^{d-1}\ketbra{ij}{ji},~
\end{equation}
and $\mathcal{N}$ is the normalization factor,
\begin{equation}\label{eq:idme}
\mathcal{N}=\sum_{i=0}^{n-1}\prod_{k=1}^{m}\big[d^{2}+(-1)^{b_{k}(i)}d\big].
\end{equation}
Note that $(\mathbbm{1}+\mathcal{F})/(d^{2}+d)$ and $(\mathbbm{1}-\mathcal{F})/(d^{2}-d)$ are PPT and NPT Werner states in a two-qudit system, respectively \cite{wern1989}.
\end{example}

It is straightforward to verify
\begin{flalign}
&(\mathbbm{1}+\mathcal{F})^{\PT}=s_{0}^{(0)}\Pi_{0}+s_{1}^{(0)}\Pi_{1},~
s_{0}^{(0)}=1+d,~s_{1}^{(0)}=1,\nonumber\\
&(\mathbbm{1}-\mathcal{F})^{\PT}=s_{0}^{(1)}\Pi_{0}+s_{1}^{(1)}\Pi_{1},~
s_{0}^{(1)}=1-d,~s_{1}^{(1)}=1,\label{eq:erp}
\end{flalign}
where $\Pi_{0}$ and $\Pi_{1}$ are mutually orthogonal projection operators such that
\begin{equation}\label{eq:opso}
\Pi_{0}=\frac{1}{d}\sum_{i,j=0}^{d-1}\ketbra{ii}{jj},~
\Pi_{1}=\mathbbm{1}-\frac{1}{d}\sum_{i,j=0}^{d-1}\ketbra{ii}{jj}.
\end{equation}
For each $i=0,\ldots,n-1$, we can see from Eq.~\eqref{eq:erp} that
\begin{equation}\label{eq:erb}
\eta_{i}\rho_{i}^{\PT}
=\frac{1}{\mathcal{N}}\sum_{\vec{a}\in\mathbb{Z}_{2}^{m}}
s_{\vec{a}}^{(i)}\Pi_{\vec{a}},
\end{equation}
where
\begin{equation}\label{eq:spam}
s_{\vec{a}}^{(i)}=\prod_{k=1}^{m}s_{a_{k}}^{(b_{k}(i))},~
\Pi_{\vec{a}}=\bigotimes_{k=1}^{m}\Pi_{a_{k}},\,
\vec{a}=(a_{1},\ldots,a_{m}).
\end{equation}

To obtain $q_{\rm G}(\mathcal{E})$, we use the measurement $\{M_{i}\}_{i=0}^{n-1}$, with $M_{0}=\mathbbm{1}^{\otimes m}$.
For each $i=0,\ldots,n-1$, we have
\begin{equation}\label{eq:emnp}
\sum_{j=0}^{n-1}\eta_{j}\rho_{j}^{\PT}M_{j}-\eta_{i}\rho_{i}^{\PT}
=\frac{1}{\mathcal{N}}\sum_{\vec{a}\in\mathbb{Z}_{2}^{m}}
(s_{\vec{a}}^{(0)}-s_{\vec{a}}^{(i)})\Pi_{\vec{a}}\succeq0,
\end{equation}
where the equality is from Eq.~\eqref{eq:erb} and the positivity is due to $s_{\vec{a}}^{(0)}\geqslant s_{\vec{a}}^{(i)}$ for all $i=0,\ldots,n-1$ and all $\vec{a}\in\mathbb{Z}_{2}^{m}$. Thus, Proposition~\ref{prop:qgm} leads us to 
\begin{equation}\label{eq:qgex2}
q_{\rm G}(\mathcal{E})=\sum_{i=0}^{n-1}\eta_{i}\Tr(\rho_{i}^{\PT}M_{i})=\eta_{0},
\end{equation}
which implies $p_{\rm L}(\mathcal{E})=q_{\rm G}(\mathcal{E})$ because
\begin{equation}\label{eq:exqg2}
\eta_{0}\leqslant p_{\rm L}(\mathcal{E})\leqslant
q_{\rm G}(\mathcal{E})=\eta_{0},
\end{equation}
where the first inequality is due to inequality~\eqref{eq:inq}, 
the second inequality is from inequality~\eqref{eq:upb},
and the equality is from Eq.~\eqref{eq:qgex2}.

\begin{figure}[!tt]
\centerline{\includegraphics{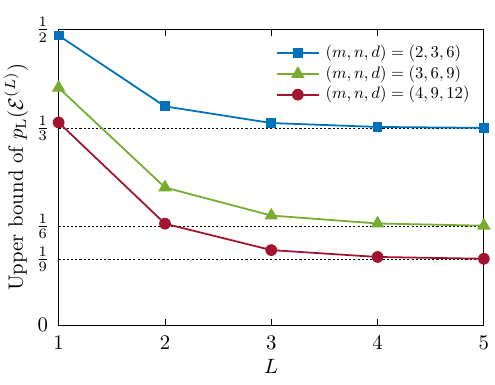}}
\caption{Decreasing of the upper bound of $p_{\rm L}(\mathcal{E}^{(L)})$ in inequality~\eqref{eq:qgel} for the ensemble $\mathcal{E}$ in Example~\ref{ex:gex} with $(m,n,d)=(2,3,6)$, $(3,6,9)$, and $(4,9,12)$.
The upper bounds decrease to $\frac{1}{n}$ exponentially fast with respect to $L$.}\label{fig:exf}
\end{figure}

We also note that
\begin{equation}\label{eq:qgq}
\eta_{0}=\frac{(d^{2}+d)^{m}}{\mathcal{N}}<\frac{(d^{2}+d)^{m}}{n(d^{2}-d)^{m}}
=\frac{1}{n}\Big(1+\frac{2}{d-1}\Big)^{m},
\end{equation}
where the inequality occurs because
\begin{equation}\label{eq:iqb}
\prod_{k=1}^{m}\big[d^{2}+(-1)^{b_{k}(i)}d\big]>(d^{2}-d)^{m}
\end{equation}
for any $i\in\{0,\ldots,n-1\}$, with $i\neq 2^{m}-1$. From Eq.~\eqref{eq:qgex2} and inequality~\eqref{eq:qgq}, we have
\begin{equation}\label{eq:qn2}
q_{\rm G}(\mathcal{E})<\frac{1}{n}\Big(1+\frac{2}{d-1}\Big)^{m}.
\end{equation}
Thus, we have $q_{\rm G}(\mathcal{E})<\frac{2}{n}$ for
\begin{equation}\label{eq:dcon}
d\geqslant\frac{\sqrt[m]{2}+1}{\sqrt[m]{2}-1}.
\end{equation}
In this case, Corollary~\ref{cor:lim} tells us that $p_{\rm L}(\mathcal{E}^{(L)})$ 
converges to $\tfrac{1}{n}$ as $L$ gets larger.
For the ensemble $\mathcal{E}$ in Eq.~\eqref{eq:enx2} with $(m,n,d)=(2,3,6)$, $(3,6,9)$, and $(4,9,12)$, Fig.~\ref{fig:exf} illustrates the decreasing of an upper bound of $p_{\rm L}(\mathcal{E}^{(L)})$, the right-hand side of inequality~\eqref{eq:qgel}, with respect to $L$.

Since $\mathbbm{1}+\mathcal{F}$ and $\mathbbm{1}-\mathcal{F}$ are orthogonal, the states of $\mathcal{E}$ in Eq.~\eqref{eq:enx2} are mutually orthogonal; therefore, the condition $p_{\rm G}(\mathcal{E})=1$ in Eq.~\eqref{eq:pqgc} holds.
From inequality~\eqref{eq:qn2}, the condition $q_{\rm G}(\mathcal{E})<\frac{2}{n}$ in Eq.~\eqref{eq:pqgc} also holds
if inequality \eqref{eq:dcon} is satisfied.
Thus, the ensemble $\mathcal{E}$ in Example~\ref{ex:gex} can be used to construct a data-hiding scheme that hides an $n$-ary classical data.

Moreover, the state $\rho_{0}$ in Eq.~\eqref{eq:enx2} is separable due to the separability of $\mathbbm{1}+\mathcal{F}$.
From this fact, we note that not all states in the ensemble satisfying condition~\eqref{eq:pqgc} necessarily need to be entangled.

\section{Discussion}\label{sec:dsc}
We considered bipartite quantum state discrimination and established a relation between a nonlocal quantum state ensemble and quantum data-hiding processing.
Using the bound $q_{\rm G}(\mathcal{E})$ on local minimum-error discrimination of bipartite quantum states, we provided a sufficient condition for a bipartite orthogonal quantum state ensemble that can be used to construct a data-hiding scheme to make classical data globally but not locally accessible, in which the data is perfectly hidden asymptotically (Corollary~\ref{cor:lim}).
Our results were illustrated with examples from multidimensional bipartite quantum systems (Examples~\ref{ex:npt} and \ref{ex:gex}).

\begin{figure}[!tt]
\centerline{\includegraphics{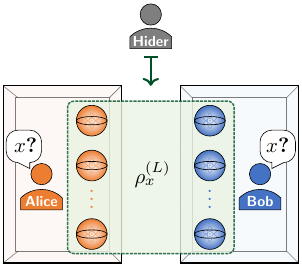}}
\caption{Data-hiding scheme using $n$ bipartite quantum states $\rho_{0}^{(L)},\ldots,\rho_{n-1}^{(L)}$ defined from a bipartite orthogonal quantum state ensemble $\mathcal{E}=\{\eta_{i},\rho_{i}\}_{i=0}^{n-1}$ with $q_{\rm G}(\mathcal{E})<\frac{2}{n}$.
To hide the classical data $x\in\{0,\ldots,n-1\}$, the hider sends the bipartite quantum state $\rho_{x}^{(L)}$ to Alice and Bob.}\label{fig:qdh3}
\end{figure}

We remark that our results can be used to hide classical data by directly encoding them into a bipartite quantum state.
For any integers $n\geqslant2$, $L\geqslant1$ and a bipartite quantum state ensemble
$\mathcal{E}=\{\eta_{i},\rho_{i}\}_{i=0}^{n-1}$ satisfying condition~\eqref{eq:pqgc}, let us consider the $n$ bipartite quantum states $\rho_{0}^{(L)},\ldots,\rho_{n-1}^{(L)}$ in Eq.~\eqref{eq:rhil}.

To hide $n$-ary classical data $x\in\mathbb{Z}_{n}$, the hider sends the state $\rho_{x}^{(L)}$ to Alice and Bob. 
This situation is illustrated in Fig.~\ref{fig:qdh3}, which is equivalent to discrimination of the states $\rho_{0}^{(L)},\ldots,\rho_{n-1}^{(L)}$ with the same prior probability $\frac{1}{n}$.

In the data-hiding scheme presented in Figs.~\ref{fig:qdh1} and \ref{fig:qdh2}, the hidden data $x$ is independent of the prepared state $\rho_{y}^{(L)}$ from the ensemble $\mathcal{E}^{(L)}$ before the information $z=x\oplus y$ is broadcast to Alice and Bob.
On the other hand, the data-hiding scheme presented in Fig.~\ref{fig:qdh3} uses only the states $\rho_{i}^{(L)}$ without the corresponding probabilities $\eta_{i}^{(L)}$; the hidden data $x$ is directly encoded into $\rho_{x}^{(L)}$.
Because Alice and Bob have no information about the hidden data $x$, the preparation of $\rho_{i}^{(L)}$ is completely random for Alice and Bob. Thus, the discrimination of the ensemble $\{\tfrac{1}{n},\rho_{i}^{(L)}\}_{i=0}^{n-1}$ occurs.

From the argument before Eq.~\eqref{eq:dpg1}, the condition $p_{\rm G}(\mathcal{E})=1$ in Eq.~\eqref{eq:pqgc} means that the states $\rho_{0}^{(L)},\ldots,\rho_{n-1}^{(L)}$ are mutually orthogonal.
Thus, we have 
\begin{equation}\label{eq:lrg}
p_{\rm G}(\{\tfrac{1}{n},\rho_{i}^{(L)}\}_{i=0}^{n-1})=1.
\end{equation}
In other words, the data $x$ is accessible when Alice and Bob can collaborate to use global measurements.

On the other hand, from Theorem~\ref{thm:plmn} in Appendix~\ref{app:thm5}, the condition $q_{\rm G}(\mathcal{E})<\frac{2}{n}$ in Eq.~\eqref{eq:pqgc} implies  that $p_{\rm L}(\{\tfrac{1}{n},\rho_{i}^{(L)}\}_{i=0}^{n-1})$ can be arbitrarily close to $\frac{1}{n}$ with the choice of an appropriately large $L$.
In other words, the globally accessible data $x$ is perfectly hidden asymptotically if Alice and Bob can use only LOCC measurements. Thus, the $n$ quantum states $\rho_{0}^{(L)},\ldots,\rho_{n-1}^{(L)}$ can be used to construct a data-hiding scheme that hides $n$-ary classical data by directly encoding them into the quantum states.

Our results tell us that any bipartite orthogonal quantum state ensemble $\mathcal{E}$ satisfying $q_{\rm G}(\mathcal{E})<\frac{2}{n}$ is useful in quantum data hiding.
Therefore, it is natural to ask whether $q_{\rm G}(\mathcal{E})<\frac{2}{n}$ is also necessary for $\mathcal{E}$ to be used in quantum data hiding. 
This question is related to the research on the operational meaning of $q_{\rm G}(\mathcal{E})$ in terms of quantum data hiding.

The results as well as the method of our paper can also be applied to construct a secure information network by extending them to multiparty quantum systems with more than two players.
In particular, it would be interesting in the future to investigate $(m,k)$-threshold $m$-player data-hiding schemes in which the hidden classical data can be recovered if $k$ or more players collaborate.

\section*{Acknowledgments}
This work was supported by a National Research Foundation of Korea(NRF) grant funded by the Korean government (Ministry of Science and ICT; Grant No. NRF2023R1A2C1007039).

\appendix
\appsection{Proofs of Theorems~\ref{thm:tsqg} and \ref{thm:ubqg}}\label{app:thm12}
\begin{proof}[Proof of Theorem~\ref{thm:tsqg}]
Let $M_{0}$ and $M_{1}$ be the projection operators onto
the non-negative and negative eigenspaces of the Hermitian operator $\eta_{0}\rho_{0}^{\PT}-\eta_{1}\rho_{1}^{\PT}$, respectively, and consider the measurement $\mathcal{M}=\{M_{i}\}_{i=0}^{1}$.
From the completeness $M_{0}+M_{1}=\mathbbm{1}$, we have
\begin{eqnarray}
\sum_{j=0}^{1}\eta_{j}\rho_{j}^{\PT}M_{j}-\eta_{i}\rho_{i}^{\PT}
&=&\eta_{0}\rho_{0}^{\PT}M_{0}
+\eta_{1}\rho_{1}^{\PT}M_{1}-\eta_{i}\rho_{i}^{\PT}(M_{0}+M_{1})
\nonumber\\
&=&(\eta_{0}\rho_{0}^{\PT}
-\eta_{i}\rho_{i}^{\PT})M_{0}
+(\eta_{1}\rho_{1}^{\PT}-\eta_{i}\rho_{i}^{\PT})M_{1}
\nonumber\\
&=&(\eta_{1-i}\rho_{1-i}^{\PT}-\eta_{i}\rho_{i}^{\PT})M_{1-i}\label{eq:tpsc}
\end{eqnarray}
for each $i=0,1$. By using the notation
\begin{equation}\label{eq:epmd}
E^{(\pm)}=\frac{1}{2}(|E|\pm E)
\end{equation}
together with the definition of $\mathcal{M}$, the last term in Eq.~\eqref{eq:tpsc} can be rewritten as
\begin{equation}\label{eq:tpsc2}
(\eta_{1-i}\rho_{1-i}^{\PT}-\eta_{i}\rho_{i}^{\PT})M_{1-i}
=(\eta_{1-i}\rho_{1-i}^{\PT}-\eta_{i}\rho_{i}^{\PT})^{(+)}
\succeq0,
\end{equation}
where the positivity is from $E^{(+)}\succeq0$ for any Hermitian operator $E$.
Equations~\eqref{eq:tpsc} and \eqref{eq:tpsc2} imply that $\mathcal{M}$ satisfies the optimality condition in Proposition~\ref{prop:qgm}.

Thus, we have
\begin{eqnarray}
q_{\rm G}(\mathcal{E})
&=&\sum_{i=0}^{1}\eta_{i}\Tr(\rho_{i}^{\PT}M_{i})\nonumber\\
&=&\frac{1}{2}\sum_{i=0}^{1}\Tr\Big(\sum_{j=0}^{1}\eta_{j}\rho_{j}^{\PT}M_{j}-\eta_{i}\rho_{i}^{\PT}+\eta_{i}\rho_{i}^{\PT}\Big)\nonumber\\
&=&\frac{1}{2}+\frac{1}{2}\sum_{i=0}^{1}\Tr(
\eta_{1-i}\rho_{1-i}^{\PT}-\eta_{i}\rho_{i}^{\PT})^{(+)}\nonumber\\
&=&\frac{1}{2}+\frac{1}{2}\Tr|\eta_{0}\rho_{0}^{\PT}-\eta_{1}\rho_{1}^{\PT}|,\label{eq:vqg}
\end{eqnarray}
where the first equality is from Proposition~\ref{prop:qgm},
the third equality is due to Eqs.~\eqref{eq:tpsc}
and \eqref{eq:tpsc2} together with $\eta_{0}+\eta_{1}=1$, and the last equality is from $E^{(-)}=(-E)^{(+)}$ and $|E|=E^{(+)}+E^{(-)}$
for any Hermitian operator $E$.
\end{proof}
\begin{proof}[Proof of Theorem~\ref{thm:ubqg}]
For the optimal measurement $\mathcal{M}=\{M_{i}\}_{i=0}^{n-1}$ providing $q_{\rm G}(\mathcal{E})$, we have
\begin{equation}
\Tr H-q_{\rm G}(\mathcal{E})=\sum_{i=0}^{n-1}\Tr[(H-\eta_{i}\rho_{i}^{\PT})M_{i}]
\geqslant0,
\end{equation}
where the equality is from $\sum_{i=0}^{n-1}M_{i}=\mathbbm{1}$ and the definition of $q_{\rm G}(\mathcal{E})$ in Eq.~\eqref{eq:qge}
and the inequality is from condition~\eqref{eq:ubc} and $M_{i}\succeq0$ for all $i=0,\ldots,n-1$. Thus, $\Tr H$ is an upper bound of $q_{\rm G}(\mathcal{E})$.
\end{proof}

\appsection{Proof of Theorem~\ref{thm:tsege}}\label{app:thm3}
From the definitions of  $\eta_{i}^{(L)}$ and $\rho_{i}^{(L)}$ in Eq.~\eqref{eq:rhil}, we can see that
\begin{eqnarray}
\eta_{i}^{(L)}\rho_{i}^{(L)}&=&\frac{1}{2}\big[(\eta_{0}\rho_{0}+\eta_{1}\rho_{1})^{\otimes L}\nonumber\\
&&+(-1)^{i}(\eta_{0}\rho_{0}-\eta_{1}\rho_{1})^{\otimes L}\big]\label{eq:er01}
\end{eqnarray}
for each $i=0,1$, which implies
\begin{eqnarray}\label{eq:tsr}
\eta_{0}^{(L)}\rho_{0}^{(L)\PT}-\eta_{1}^{(L)}\rho_{1}^{(L)\PT}
=(\eta_{0}\rho_{0}^{\PT}-\eta_{1}\rho_{1}^{\PT})^{\otimes L}.
\end{eqnarray}
For the two-state ensemble $\mathcal{E}^{(L)}=\{\eta_{i}^{(L)},\rho_{i}^{(L)}\}_{i=0}^{1}$, we have
\begin{eqnarray}\label{eq:tqer}
q_{\rm G}(\mathcal{E}^{(L)})
&=&\frac{1}{2}+\frac{1}{2}\Tr\big|\eta_{0}^{(L)}\rho_{0}^{(L)\PT}-\eta_{1}^{(L)}\rho_{1}^{(L)\PT}\big|\nonumber\\
&=&
\frac{1}{2}+\frac{1}{2}\Tr\big|(\eta_{0}\rho_{0}^{\PT}-\eta_{1}\rho_{1}^{\PT})^{\otimes L}\big|
\nonumber\\
&=&\frac{1}{2}+\frac{1}{2}\big(\Tr\big|\eta_{0}\rho_{0}^{\PT}-\eta_{1}\rho_{1}^{\PT}\big|\big)^{L},
\end{eqnarray}
where the first equality is from Theorem~\ref{thm:tsqg},
the second equality is from Eq.~\eqref{eq:tsr}, and
the last equality is due to $|E^{\otimes L}|=|E|^{\otimes L}$ for any Hermitian operator $E$. 

Now, let us assume $p_{\rm L}(\mathcal{E})=q_{\rm G}(\mathcal{E})$. 
To prove Eq.~\eqref{eq:plqgl}, it is enough to show 
\begin{equation}\label{eq:pllw}
p_{\rm L}(\mathcal{E}^{(L)})\geqslant q_{\rm G}(\mathcal{E}^{(L)})
\end{equation}
because $p_{\rm L}(\mathcal{E}^{(L)})\leqslant q_{\rm G}(\mathcal{E}^{(L)})$ due to inequality~\eqref{eq:upb}.

Let $\mathcal{M}=\{M_{i}\}_{i=0}^{1}$ be a LOCC measurement realizing $p_{\rm L}(\mathcal{E})$, and consider the measurement $\mathcal{M}^{(L)}=\{M_{i}^{(L)}\}_{i=0}^{1}$ with
\begin{equation}
M_{i}^{(L)}=\sum_{\substack{\vec{c}\in\mathbb{Z}_{2}^{L}\\ \omega_{2}(\vec{c})=i}}\bigotimes_{l=1}^{L}M_{c_{l}}.
\end{equation}
Note that $\mathcal{M}^{(L)}$ is a LOCC measurement because it can be realized by performing the same LOCC measurement $\mathcal{M}$ for each of $L$ subsystems. 

From the definition of $\omega_{2}(\vec{c})$ in Eq.~\eqref{eq:mns}, we have
\begin{equation}
M_{i}^{(L)}=\frac{1}{2}\big[(M_{0}+M_{1})^{\otimes L}
+(-1)^{i}(M_{0}-M_{1})^{\otimes L}\big]\label{eq:wm01}
\end{equation}
for all $i=0,1$.
Now, we note that the success probability of discriminating the states from $\mathcal{E}^{(L)}$
using $\mathcal{M}^{(L)}$ is equal to $q_{\rm G}(\mathcal{E}^{(L)})$ because
\begin{eqnarray}
\sum_{i=0}^{1}\eta_{i}^{(L)}\Tr(\rho_{i}^{(L)}M_{i}^{(L)})
&=&\frac{1}{2}\Tr\Big[(\eta_{0}\rho_{0}+\eta_{1}\rho_{1})^{\otimes L}(M_{0}+M_{1})^{\otimes L}\Big]
+\frac{1}{2}\Tr\Big[(\eta_{0}\rho_{0}-\eta_{1}\rho_{1})^{\otimes L}(M_{0}-M_{1})^{\otimes L}\Big]\nonumber\\
&=&\frac{1}{2}\Tr\Big[\sum_{i=0}^{1}\eta_{i}\rho_{i}\Big]^{\otimes L}+\frac{1}{2}\Tr\Big[\sum_{i=0}^{1}\eta_{i}\rho_{i}(2M_{i}-\mathbbm{1})\Big]^{\otimes L}\nonumber\\
&=&\frac{1}{2}+\frac{1}{2}\Big[2\sum_{i=0}^{1}\eta_{i}\Tr(\rho_{i}M_{i})-1\Big]^{L}
=\frac{1}{2}+\frac{1}{2}\big(2q_{\rm G}(\mathcal{E})-1\big)^{L}
=q_{\rm G}(\mathcal{E}^{(L)}),\label{eq:apcg}
\end{eqnarray}
where the first equality is from Eqs.~\eqref{eq:er01} and \eqref{eq:wm01}, the second equality is due to $M_{0}+M_{1}=\mathbbm{1}$, the third equality is from $\sum_{i=0}^{1}\Tr(\eta_{i}\rho_{i})=1$, the fourth equality follows from the assumption of $\mathcal{M}$ and $p_{\rm L}(\mathcal{E})=q_{\rm G}(\mathcal{E})$, and the last equality is from Theorem~\ref{thm:tsqg} together with Eq.~\eqref{eq:tqer}.
The success probability in Eq.~\eqref{eq:apcg} is obviously a lower bound of $p_{\rm L}(\mathcal{E}^{(L)})$; therefore, inequality~\eqref{eq:pllw} holds.

\appsection{Proof of Theorem~\ref{thm:lqge}}\label{app:thm4}
Let $\mathcal{M}=\{M_{i}\}_{i=0}^{n-1}$ be a measurement providing $q_{\rm G}(\mathcal{E})$, and consider
\begin{equation}\label{eq:gam}
\Gamma=\sum_{i=0}^{n-1}\eta_{i}\rho_{i}^{\PT},~
H=\sum_{i=0}^{n-1}\eta_{i}\rho_{i}^{\PT}M_{i}.
\end{equation}
From Proposition~\ref{prop:qgm}, we have
\begin{equation}\label{eq:hept0}
H-\eta_{i}\rho_{i}^{\PT}=\sum_{j=0}^{n-1}\eta_{j}\rho_{j}^{\PT}M_{j}-\eta_{i}\rho_{i}^{\PT}\succeq0 
\end{equation}
for all $i=0,\ldots,n-1$. Let 
\begin{equation}\label{eq:wid}
W_{i}=H-\eta_{i}\rho_{i}^{\PT}\succeq0
\end{equation}
for each $i=0,\ldots,n-1$, and define the positive semidefinite operator 
\begin{equation}\label{eq:wri}
W_{\vec{c}}=\bigotimes_{l=1}^{L}W_{c_{l}}\succeq0
\end{equation}
for each $\vec{c}=(c_{1},\ldots,c_{L})\in\mathbb{Z}_{n}^{L}$ with the $l$th coordinate $c_{l}$.

We first note that for each $i=0,\ldots,n-1$,
\begin{equation}\label{eq:hlii}
\sum_{\substack{\vec{c}\in\mathbb{Z}_{n}^{L}\\ \omega_{n}(\vec{c})= i}}W_{\vec{c}}
+\sum_{\substack{\vec{c}\in\mathbb{Z}_{n}^{L}\\ \omega_{n}(\vec{c})\neq i}}W_{\vec{c}}
=\Big(\sum_{j=0}^{n-1}W_{j}\Big)^{\otimes L}
=(nH-\Gamma)^{\otimes L},
\end{equation}
where the first equality is from the definition of $W_{\vec{c}}$ in Eq.~\eqref{eq:wri} and the last equality is from Eqs.~\eqref{eq:gam} and \eqref{eq:wid}.

\begin{lemma}\label{lem:tpl}
For a bipartite quantum state ensemble $\mathcal{E}=\{\eta_{i},\rho_{i}\}_{i=0}^{n-1}$ and for any positive integer $L$,
\begin{equation}\label{eq:hlde}
\sum_{\substack{\vec{c}\in\mathbb{Z}_{n}^{L}\\ \omega_{n}(\vec{c})= i}}W_{\vec{c}}-(-1)^{L}\eta_{i}^{(L)}\rho_{i}^{(L)\PT}
=\frac{1}{n}(nH-\Gamma)^{\otimes L}-(-1)^{L}\frac{1}{n}\Gamma^{\otimes L},
\end{equation}
where $H$ and $\Gamma$ are from Eq.~\eqref{eq:gam},
$W_{\vec{c}}$ is from Eq.~\eqref{eq:wri},
and $\eta_{i}^{(L)}$ and $\rho_{i}^{(L)}$ are from Eq.~\eqref{eq:rhil}.
\end{lemma}

\begin{proof}[Proof of Lemma of \ref{lem:tpl}]
We first note the following relation
\begin{equation}\label{eq:abl}
\bigotimes_{l=1}^{L}A_{l}-\bigotimes_{l=1}^{L}B_{l}
=\sum_{k=1}^{L}A_{1}\otimes\cdots\otimes A_{L-k}\otimes(A_{L-k+1}-B_{L-k+1})
\otimes B_{L-k+2}\otimes\cdots\otimes B_{L}
\end{equation}
for any operators $A_{l}$ and $B_{l}$.
For each $i=0,\ldots,n-1$, we can see that
\begin{eqnarray}\label{eq:wnh1}
\sum_{\substack{\vec{c}\in\mathbb{Z}_{n}^{L}\\ \omega_{n}(\vec{c})= i}}W_{\vec{c}}-(-1)^{L}\eta_{i}^{(L)}\rho_{i}^{(L)\PT}
&=&\sum_{\substack{\vec{c}\in\mathbb{Z}_{n}^{L}\\ \omega_{n}(\vec{c})= i}}\Bigg[\bigotimes_{l=1}^{L}W_{c_{l}}-\bigotimes_{l=1}^{L}(-\eta_{c_{l}}\rho_{c_{l}}^{\PT})\Bigg]\nonumber\\
&=&\sum_{\substack{\vec{c}\in\mathbb{Z}_{n}^{L}\\ \omega_{n}(\vec{c})= i}}\sum_{k=1}^{L}
W_{c_{1}}\otimes\cdots\otimes W_{c_{L-k}}\otimes H
\otimes(-\eta_{c_{L-k+2}}\rho_{c_{L-k+2}}^{\PT})\otimes\cdots\otimes(-\eta_{c_{L}}\rho_{c_{L}}^{\PT})\nonumber\\
&=&\sum_{k=1}^{L}(\sum_{j\in\mathbb{Z}_{n}}W_{j})^{\otimes L-k}\otimes H
\otimes\Big[\sum_{j\in\mathbb{Z}_{n}}(-\eta_{j}\rho_{j}^{\PT})\Big]^{\otimes k-1},
\end{eqnarray}
where the first equality is due to Eqs.~\eqref{eq:rhil} and \eqref{eq:wri} and the second equality is from Eq.~\eqref{eq:abl} together with Eq.~\eqref{eq:wid}. The last equality in Eq.~\eqref{eq:wnh1} holds because
\begin{eqnarray}
\{\vec{c}\in\mathbb{Z}_{n}^{L}\,|\,
\omega_{n}(\vec{c})=i\}
=\{(c_{1},\ldots,c_{L})&|&c_{1}\in\mathbb{Z}_{n},\ldots,
c_{L-k}\in\mathbb{Z}_{n},\,
c_{L-k+2}\in\mathbb{Z}_{n},\ldots,c_{L}\in\mathbb{Z}_{n},\nonumber\\
&&c_{L-k+1}=i-\mbox{$\sum_{l=1,l\neq L-k+1}^{L}c_{l}$}
~\mbox{mod}~n\,\}
\end{eqnarray}
for all $i=0,\ldots,n-1$ and all $k=1,\ldots,L$.

The last term in Eq.~\eqref{eq:wnh1} can be rewritten as
\begin{eqnarray}
\sum_{k=1}^{L}(\sum_{j\in\mathbb{Z}_{n}}W_{j})^{\otimes L-k}\otimes H
\otimes\Big[\sum_{j\in\mathbb{Z}_{n}}(-\eta_{j}\rho_{j}^{\PT})\Big]^{\otimes k-1}
&=&\sum_{k=1}^{L}(n H-\Gamma)^{\otimes L-k}\otimes H
\otimes(-\Gamma)^{\otimes k-1} \nonumber\\
&=&\frac{1}{n}\sum_{k=1}^{L}(n H-\Gamma)^{\otimes L-k}\otimes \big[nH-\Gamma-(-\Gamma)\big]
\otimes(-\Gamma)^{\otimes k-1} \nonumber\\
&=&\frac{1}{n}\Big[(nH-\Gamma)^{\otimes L}-(-\Gamma)^{\otimes L}\Big],
\label{eq:wnh2}
\end{eqnarray}
where the first equality is from Eqs.~\eqref{eq:gam} and \eqref{eq:wid} and the last equality is due to Eq.~\eqref{eq:abl}.
Equations~\eqref{eq:wnh1} and \eqref{eq:wnh2} lead us to Eq.~\eqref{eq:hlde}.
\end{proof}

Now, let us consider
\begin{equation}\label{eq:hnl}
H^{(L)}=
\begin{dcases}
\tfrac{1}{n}\Gamma^{\otimes L}+\tfrac{1}{n}(n H-\Gamma)^{\otimes L},&L:\mathrm{odd},\\
\tfrac{1}{n}\Gamma^{\otimes L}+\tfrac{n-1}{n}(n H-\Gamma)^{\otimes L},&L:\mathrm{even}.
\end{dcases}
\end{equation}
From Lemma~\ref{lem:tpl} together with Eq.~\eqref{eq:hlii}, we have
\begin{equation}
H^{(L)}-\eta_{i}^{(L)}\rho_{i}^{(L)\PT}=
\begin{dcases}
\sum_{\substack{\vec{c}\in\mathbb{Z}_{n}^{L}\\ \omega_{n}(\vec{c})= i}}W_{\vec{c}},&L:\mathrm{odd},\\
\sum_{\substack{\vec{c}\in\mathbb{Z}_{n}^{L}\\ \omega_{n}(\vec{c})\neq i}}W_{\vec{c}},&L:\mathrm{even},
\end{dcases}
\label{eq:hldc}
\end{equation}
for all $i=0,\ldots,n-1$.

Due to the non-negativity of $W_{\vec{c}}$ in \eqref{eq:wri}, Eq.~\eqref{eq:hldc} implies
\begin{equation}\label{eq:hlic}
H^{(L)}-\eta_{i}^{(L)}\rho_{i}^{(L)\PT}\succeq0.
\end{equation}
In other words, $H^{(L)}$ satisfies the condition of Theorem~\ref{thm:ubqg} for the ensemble $\mathcal{E}^{(L)}$ in Eq.~\eqref{eq:eld}; therefore, its trace provides an upper bound of $q_{\rm G}(\mathcal{E}^{(L)})$,
\begin{equation}\label{eq:hlqg}
q_{\rm G}(\mathcal{E}^{(L)})\leqslant \Tr H^{(L)}.
\end{equation}
As $\Tr\Gamma=1$ and $\Tr H=q_{\rm G}(\mathcal{E})$, we have
\begin{eqnarray}
\Tr H^{(L)}
&=&\begin{dcases}
\tfrac{1}{n}+\tfrac{1}{n}\big[n q_{\rm G}(\mathcal{E})-1\big]^{L},&L:\mathrm{odd},\\
\tfrac{1}{n}+\tfrac{n-1}{n}\big[n q_{\rm G}(\mathcal{E})-1\big]^{L},&L:\mathrm{even},
\end{dcases}
\nonumber\\[3mm]
&\leqslant&\frac{1}{n}+\frac{n-1}{n}\big[n q_{\rm G}(\mathcal{E})-1\big]^{L},\label{eq:trhl}
\end{eqnarray}
where the equality is due to Eq.~\eqref{eq:hnl} and the inequality is from the non-negativity of $n q_{\rm G}(\mathcal{E})-1$. Inequalities~\eqref{eq:hlqg} and \eqref{eq:trhl} lead us to inequality~\eqref{eq:qgel}, which completes the proof.

\appsection{Convergence of $\bm{p_{\rm L}(\{\frac{1}{n},\rho_{i}^{(L)}\}_{i=0}^{n-1})}$}\label{app:thm5}
\begin{theorem}\label{thm:plmn}
For a bipartite quantum state ensemble $\mathcal{E}=\{\eta_{i},\rho_{i}\}_{i=0}^{n-1}$ with
$q_{\rm G}(\mathcal{E})<\frac{2}{n}$,
\begin{equation}\label{eq:plmn}
\lim_{L\rightarrow\infty}p_{\rm L}(\{\tfrac{1}{n},\rho_{i}^{(L)}\}_{i=0}^{n-1})=\frac{1}{n}.
\end{equation}
\end{theorem}
\begin{proof}
For the quantum states $\rho_{0}^{(L)},\ldots,\rho_{n-1}^{(L)}$ in Eq.~\eqref{eq:rhil}, we will show that
\begin{equation}\label{eq:atie}
p_{\rm L}(\{\tfrac{1}{n},\rho_{i}^{(L)}\}_{i=0}^{n-1})
\leqslant\frac{1}{n}+\frac{(n-1)(n^{2}-n+2)}{2n}\big[n q_{\rm G}(\mathcal{E})-1\big]^{L}.
\end{equation}
Because $q_{\rm G}(\mathcal{E})<\frac{2}{n}$,
the right-hand side of inequality~\eqref{eq:atie} converges to $\tfrac{1}{n}$ as $L$ gets larger.
Moreover, $p_{\rm L}(\{\tfrac{1}{n},\rho_{i}^{(L)}\}_{i=0}^{n-1})$ is bounded below by $\frac{1}{n}$ from inequality~\eqref{eq:inq}. Therefore, inequality~\eqref{eq:atie} leads us to
Eq.~\eqref{eq:plmn}.

Let $\mathcal{M}=\{M_{i}\}_{i=0}^{n-1}$ be a LOCC measurement realizing $p_{\rm L}(\{\tfrac{1}{n},\rho_{i}^{(L)}\}_{i=0}^{n-1})$.
From the definition of $p_{\rm L}(\mathcal{E}^{(L)})$ in Eq.~\eqref{eq:ple}, we have
\begin{equation}
p_{\rm L}(\{\tfrac{1}{n},\rho_{i}^{(L)}\}_{i=0}^{n-1})
-p_{\rm L}(\mathcal{E}^{(L)})
=\sum_{i=0}^{n-1}\frac{1}{n}\Tr(\rho_{i}^{(L)}M_{i})
-p_{\rm L}(\mathcal{E}^{(L)})
\leqslant\sum_{i=0}^{n-1}\Big(\frac{1}{n}-\eta_{i}^{(L)}\Big)\Tr(\rho_{i}^{(L)}M_{i}).\label{eq:ppd}
\end{equation}

Since $0\leqslant\Tr(\rho_{i}^{(L)}M_{i})\leqslant 1$ for any $i\in\{0,\ldots,n-1\}$ and $r\leqslant\frac{1}{2}(|r|+r)$ for any real number $r$, we have
\begin{equation}\label{eq:phe}
\sum_{i=0}^{n-1}\Big(\frac{1}{n}-\eta_{i}^{(L)}\Big)\Tr(\rho_{i}^{(L)}M_{i})
\leqslant\sum_{i=0}^{n-1}\frac{1}{2}\Bigg[\Big|\frac{1}{n}-\eta_{i}^{(L)}\Big|+\Big(\frac{1}{n}-\eta_{i}^{(L)}\Big)\Bigg]
=\frac{1}{2}\sum_{i=0}^{n-1}\Big|\frac{1}{n}-\eta_{i}^{(L)}\Big|,
\end{equation}
where the last equality holds from $\sum_{i=0}^{n-1}\eta_{i}^{(L)}=1$.
Thus, we have
\begin{eqnarray}
p_{\rm L}(\{\tfrac{1}{n},\rho_{i}^{(L)}\}_{i=0}^{n-1})
&\leqslant&p_{\rm L}(\mathcal{E}^{(L)})+\frac{1}{2}\sum_{i=0}^{n-1}\Big|\frac{1}{n}-\eta_{i}^{(L)}\Big|\nonumber\\
&\leqslant&q_{\rm G}(\mathcal{E}^{(L)})+\frac{1}{2}\sum_{i=0}^{n-1}\Big|\frac{1}{n}-\eta_{i}^{(L)}\Big|\nonumber\\
&\leqslant&\frac{1}{n}+\frac{n-1}{n}\big[n q_{\rm G}(\mathcal{E})-1\big]^{L}+\frac{1}{2}\sum_{i=0}^{n-1}\Big|\frac{1}{n}-\eta_{i}^{(L)}\Big|,
\label{eq:tinq}
\end{eqnarray}
where the first inequality is due to inequalities~\eqref{eq:ppd} and \eqref{eq:phe}, 
the second inequality is from inequality~\eqref{eq:upb}, 
and the last inequality is from Theorem~\ref{thm:lqge}.

To obtain an upper bound of $\frac{1}{2}\sum_{i=0}^{n-1}|\frac{1}{n}-\eta_{i}^{(L)}|$, let us first consider
\begin{equation}\label{eq:minhi}
\max\{\eta_{0}^{(L)},\ldots,\eta_{n-1}^{(L)}\}
\leqslant p_{\rm L}(\mathcal{E}^{(L)})\leqslant q_{\rm G}(\mathcal{E}^{(L)})
\leqslant\frac{1}{n}+\frac{n-1}{n}\big[n q_{\rm G}(\mathcal{E})-1\big]^{L},
\end{equation}
where the first inequality is from inequality~\eqref{eq:inq},
the second inequality is due to inequality~\eqref{eq:upb},
and the last inequality is from Theorem~\ref{thm:lqge}.
Moreover, we have
\begin{equation}\label{eq:maxhi}
\min\{\eta_{0}^{(L)},\ldots,\eta_{n-1}^{(L)}\}
\geqslant1-(n-1)\max\{\eta_{0}^{(L)},\ldots,\eta_{n-1}^{(L)}\}
\geqslant\frac{1}{n}-\frac{(n-1)^{2}}{n}\big[n q_{\rm G}(\mathcal{E})-1\big]^{L},
\end{equation}
where the first inequality is because $\{\eta_{i}^{(L)}\}_{i=0}^{n-1}$ is a probability distribution, and the second inequality follows from Eq.~\eqref{eq:minhi}.

From inequalities \eqref{eq:minhi} and \eqref{eq:maxhi}, we have
\begin{equation}\label{eq:abel}
\Big|\frac{1}{n}-\eta_{i}^{(L)}\Big|\leqslant\frac{(n-1)^{2}}{n}\big[n q_{\rm G}(\mathcal{E})-1\big]^{L}
\end{equation}
for each $i=0,\ldots,n-1$, which implies
\begin{equation}\label{eq:ltub}
\frac{1}{2}\sum_{i=0}^{n-1}\Big|\frac{1}{n}-\eta_{i}^{(L)}\Big|\leqslant
\frac{(n-1)^{2}}{2}\big[n q_{\rm G}(\mathcal{E})-1\big]^{L}.
\end{equation}
Inequalities~\eqref{eq:tinq} and \eqref{eq:ltub} lead us to inequality~\eqref{eq:atie}.
\end{proof}



\begin{thebibliography}{99}
\bibitem{pptp}  
The PPT property does not depend on the choice of basis or the subsystem to be transposed. 
For simplicity, we consider the standard basis and Bob's system throughout this paper.
\end{thebibliography}
\end{document}